\documentclass[11pt,a4paper]{article}

\usepackage[utf8]{inputenc}
\usepackage[T1]{fontenc}
\usepackage{amsmath,amssymb,amsthm}
\usepackage{booktabs}
\usepackage{graphicx}
\usepackage{hyperref}
\usepackage[margin=1in]{geometry}
\usepackage{natbib}
\usepackage{xcolor}
\usepackage{algorithm}
\usepackage{algpseudocode}
\usepackage{listings}
\usepackage{caption}
\usepackage{subcaption}
\usepackage{multirow}
\usepackage{tabularx}
\usepackage{microtype}

\newtheorem{theorem}{Theorem}
\newtheorem{definition}{Definition}

\title{The Global Representativeness Index:\\A Total Variation Distance Framework for Measuring\\Demographic Fidelity in Survey Research}

\author{Evan Hadfield\thanks{Correspondence: \texttt{evan@cip.org}} \and Andrew Konya}

\date{}

\begin{document}

\maketitle

\begin{abstract}
Global survey research increasingly informs high-stakes decisions in AI governance and cross-cultural policy, yet no standardized metric quantifies how well a sample's demographic composition matches its target population. Response rates and demographic quotas---the prevailing proxies for sample quality---measure effort and coverage but not distributional fidelity. This paper introduces the Global Representativeness Index (GRI), a framework grounded in Total Variation Distance that scores any survey sample against population benchmarks across multiple demographic dimensions on a $[0, 1]$ scale. Validation on seven waves of the Global Dialogues survey ($N = 7{,}500$ across 60+ countries) finds fine-grained demographic GRI scores of only 0.33--0.36---roughly 43\% of the theoretical maximum at that sample size. Cross-validation on the World Values Survey (seven waves, $N = 403{,}000$), Afrobarometer Round~9 ($N = 53{,}000$), and Latinobar\'{o}metro ($N = 19{,}000$) reveals that even large probability surveys score below 0.22 on fine-grained global demographics when country coverage is limited. The GRI connects to classical survey statistics through the design effect; both metrics are recommended as a minimum summary of sample quality, since GRI quantifies demographic distance symmetrically while effective $N$ captures the asymmetric inferential cost of underrepresentation. The paper also introduces the Strategic Representativeness Index (SRI) for optimizing sampling design. The framework is released as an open-source Python library with UN and Pew Research Center population benchmarks, applicable to survey research, machine learning dataset auditing, and AI evaluation benchmarks.
\end{abstract}


\section{Introduction}

\subsection{The Stakes of Non-Representative Data}

When the European Union drafted the AI Act \citep{eu2024aiact}, it drew on public opinion research to calibrate risk categories. When UNESCO published its Recommendation on the Ethics of Artificial Intelligence---the first global normative instrument on AI \citep{unesco2021ai}---the consultation process claimed to represent ``global perspectives.'' But whose perspectives? A survey that oversamples urban, English-speaking, highly-educated respondents from a handful of countries does not represent the world, regardless of how many countries appear in its sample frame. The gap between claimed and actual representativeness shapes which voices inform regulation, which cultural values get encoded into AI systems, and which populations bear the costs of policies designed without their input.

This problem extends beyond AI governance. The World Values Survey \citep{haerpfer2022wvs}, Afrobarometer, Latinobar\'{o}metro, and other major cross-national instruments invest heavily in sampling design, yet each faces the fundamental challenge of measuring how closely their achieved samples match the populations they claim to represent. Response rates---once the gold standard of survey quality---have declined precipitously and measure participation rather than demographic fidelity \citep{brick2013explaining}. Demographic quotas ensure minimum representation of specified groups but say nothing about whether the \emph{joint distribution} of characteristics in the sample mirrors the population. Post-stratification weights can adjust for known imbalances, but they correct analysis rather than measuring the underlying problem.

What the field lacks is a formal, reproducible metric that answers a simple question: \emph{How well does this sample's demographic composition match the target population?}

\subsection{The Measurement Gap}

The absence of a standardized representativeness metric creates three practical problems. First, researchers cannot objectively compare the demographic quality of different surveys or different waves of the same survey. Second, survey designers lack quantitative targets for sampling---they can set quotas for individual demographics but cannot optimize for the joint distribution across multiple dimensions. Third, consumers of survey research---policymakers, journalists, the public---have no way to assess claims of ``global'' or ``representative'' sampling beyond trusting the methodology section.

R-indicators \citep{schouten2009indicators} measure the variation in response propensities across population subgroups, requiring auxiliary population data to model these propensities. R-indicators answer a related but distinct question: ``How uniform is the response mechanism?'' rather than ``How closely does the achieved sample match the population distribution?'' The GRI complements R-indicators by directly measuring the distributional outcome regardless of the mechanism that produced it---a distinction that matters particularly for non-probability samples where response propensities are undefined. Balance metrics from the causal inference literature \citep{rubin2001using,stuart2010matching} assess covariate balance between treatment and control groups but are designed for experimental rather than survey contexts.

\subsection{Contributions}

This paper makes four contributions:
\begin{enumerate}
    \item \textbf{A formal mathematical framework} for measuring survey representativeness based on Total Variation Distance, producing an interpretable score bounded on $[0, 1]$ with clear properties.
    \item \textbf{A multi-dimensional approach} that evaluates representativeness across three benchmark dimensions simultaneously---Country $\times$ Gender $\times$ Age, Country $\times$ Religion, and Country $\times$ Urban/Rural Environment---using authoritative population data from the United Nations and Pew Research Center.
    \item \textbf{Empirical validation} through application to seven waves of the Global Dialogues survey ($N \approx 1{,}000$ per wave), seven waves of the World Values Survey ($N = 9{,}000$--$85{,}000$ per wave), and regional surveys (Afrobarometer, Latinobar\'{o}metro), including Monte Carlo simulation of maximum achievable scores and efficiency analysis.
    \item \textbf{An open-source Python library} (\texttt{gri}) implementing the complete framework.
\end{enumerate}

\subsection{Scope and Normative Commitments}

A representativeness metric necessarily embeds normative choices. The GRI, as formulated with global population benchmarks, adopts a specific normative position: \emph{one person, one unit of representativeness}. A survey is maximally representative when its demographic composition mirrors the world's population, weighting each person equally regardless of nationality, political influence, or institutional context.

This is a defensible default for research claiming to capture ``global perspectives.'' But it is not the only defensible choice. An AI governance survey might reasonably weight countries by their AI capability or regulatory influence; a climate adaptation survey might weight by climate vulnerability. The GRI framework accommodates such choices: researchers can substitute custom population benchmarks that reflect their specific representativeness goals.

We also note an important distinction: the GRI measures \emph{marginal distributional distance}---how closely the sample's demographic composition matches the target population. It does not directly measure inferential quality: a sample with low GRI can produce unbiased estimates if appropriate post-stratification weights are applied, and a sample with high GRI can still produce biased estimates if respondents within each demographic cell are non-randomly selected. The GRI measures the \emph{input} to inference---the demographic composition of the raw sample---not the quality of the estimates that emerge after weighting and adjustment.

\section{Beyond Response Rates: The Case for Distributional Metrics}

\subsection{Classical Foundations and Their Limits}

Survey sampling theory, formalized by \citet{neyman1934two} and extended by \citet{kish1965survey} and \citet{horvitz1952generalization}, provides rigorous frameworks for designing probability samples and producing unbiased estimators. For global surveys, the conditions for probability sampling rarely hold simultaneously: no complete sampling frame exists for the world's population, and multi-country probability sampling requires coordinated fieldwork across vastly different national contexts.

Non-probability approaches---online panels, snowball sampling, convenience samples---are increasingly common precisely because they are feasible at scale. Such designs sacrifice the theoretical guarantees of probability sampling in exchange for breadth and speed. The question of \emph{how representative the resulting sample actually is} becomes correspondingly more urgent.

\subsection{Representativeness Metrics in Current Practice}

Response rates bear little systematic relationship to nonresponse bias \citep{groves2006nonresponse}, a finding reinforced by meta-analyses across survey contexts \citep{groves2008impact}. R-indicators \citep{schouten2009indicators} measure how much response propensities vary across population subgroups. Like the GRI, R-indicators require auxiliary population data. The key conceptual difference is that R-indicators measure \emph{mechanism} (variation in response propensities) while the GRI measures \emph{outcome} (distance between achieved and target distributions). R-indicators are well-suited to probability surveys where response propensities are meaningful; the GRI applies equally to probability and non-probability samples.

\subsection{Total Variation Distance as a Foundation}

Total Variation Distance (TVD) is the natural metric for comparing discrete probability distributions. For two distributions $P$ and $Q$ defined over the same finite set:
\begin{equation}\label{eq:tvd}
    \text{TVD}(P, Q) = \frac{1}{2} \sum_{i} |p_i - q_i|
\end{equation}

TVD has several desirable properties: it is bounded ($0 \leq \text{TVD} \leq 1$), symmetric, interpretable (the fraction of probability mass that must be moved to transform one distribution into the other), non-parametric, and decomposable (each stratum's contribution is identifiable).

Alternative metrics have drawbacks for this application. KL divergence is asymmetric, undefined when $q_i = 0$, and unbounded. The Hellinger distance, $H(P,Q) = \frac{1}{\sqrt{2}}\sqrt{\sum_i (\sqrt{p_i} - \sqrt{q_i})^2}$, is symmetric and bounded but less directly interpretable. The chi-squared distance is sensitive to small expected counts---precisely the scenario that arises frequently in global benchmarks.

We opt for a TVD + Diversity Score decomposition over a single Hellinger-based index for three reasons: (1) \textbf{interpretability}---a TVD of 0.3 means ``30\% of mass is misallocated,'' while Hellinger 0.3 has no simple prose interpretation; (2) \textbf{decomposability}---TVD decomposes additively into per-stratum contributions, enabling segment-level diagnostics; and (3) \textbf{separation of concerns}---distributional fidelity and coverage are conceptually distinct aspects that practitioners need to evaluate independently. TVD and Hellinger are formally related: $H^2(P,Q) \leq \text{TVD}(P,Q) \leq H(P,Q)\sqrt{2}$, so the choice is one of transparency, not mathematical necessity.

\section{Methodology}

\subsection{The Global Representativeness Index}

\subsubsection{Formal Definition}

\begin{definition}
Let $\mathcal{S} = \{s_1, s_2, \ldots, s_K\}$ denote the set of $K$ demographic strata. For a survey sample of size $N$ and a reference population, define $p_i$ as the sample proportion in stratum $i$ and $q_i$ as the population proportion. The \textbf{Global Representativeness Index} is:
\begin{equation}\label{eq:gri}
    \text{GRI} = 1 - \frac{1}{2} \sum_{i=1}^{K} |p_i - q_i|
\end{equation}
Equivalently, $\text{GRI} = 1 - \text{TVD}(P, Q)$.
\end{definition}

\subsubsection{Properties}

\begin{theorem}[Boundedness]\label{thm:bounds}
For any two discrete probability distributions $P$ and $Q$ over $K$ categories, $0 \leq \text{GRI}(P, Q) \leq 1$.
\end{theorem}
\begin{proof}
Since $\text{GRI} = 1 - \text{TVD}$, it suffices to show $0 \leq \text{TVD} \leq 1$. The lower bound follows from non-negativity of absolute values. For the upper bound: $\sum_i |p_i - q_i| \leq \sum_i (p_i + q_i) = 2$, so $\text{TVD} \leq 1$. Equality holds when $P$ and $Q$ have disjoint support.
\end{proof}

\begin{theorem}[Monotonicity]\label{thm:mono}
Let $P$ be a sample distribution and $Q$ the target. Transferring mass $\delta > 0$ from an over-represented stratum $j$ ($p_j > q_j$) to an under-represented stratum $k$ ($p_k < q_k$), with $\delta \leq \min(p_j - q_j, q_k - p_k)$, strictly increases the GRI.
\end{theorem}
\begin{proof}
The change in TVD is $\Delta\text{TVD} = \frac{1}{2}(|p_j - \delta - q_j| + |q_k - p_k - \delta| - |p_j - q_j| - |q_k - p_k|) = -\delta < 0$.
\end{proof}

\subsubsection{Interpretation Scale}

We propose interpretation guidelines calibrated against Monte Carlo simulations of achievable scores (Table~\ref{tab:interpretation}).

\begin{table}[ht]
\centering
\caption{GRI interpretation guidelines.}
\label{tab:interpretation}
\begin{tabular}{@{}lll@{}}
\toprule
GRI Score & Interpretation & Meaning \\
\midrule
$0.8$--$1.0$ & Excellent & $<20\%$ of mass misallocated \\
$0.6$--$0.8$ & Good & $20$--$40\%$ misallocation \\
$0.4$--$0.6$ & Moderate & $40$--$60\%$ misallocation \\
$0.0$--$0.4$ & Poor & $>60\%$ misallocation \\
\bottomrule
\end{tabular}
\end{table}

\subsection{The Diversity Score}

The GRI measures distributional fidelity; the Diversity Score captures \emph{coverage}---how many population strata the sample reaches at all.

\begin{definition}
Let $N$ be the sample size and $X = 1/N$ the relevance threshold. A stratum is \textbf{relevant} if $q_i > X$ and \textbf{represented} if $p_i > 0$. The Diversity Score is:
\begin{equation}\label{eq:diversity}
    \text{DiversityScore} = \frac{|\{i : p_i > 0 \text{ and } q_i > X\}|}{|\{i : q_i > X\}|}
\end{equation}
\end{definition}

The threshold $X = 1/N$ identifies strata where we would expect at least one observation in a perfectly proportional sample. Under the Poisson approximation, the probability of observing at least one member is $1 - e^{-1} \approx 0.632$ for relevant strata.

\subsection{Multi-Dimensional Scorecard}

Three primary dimensions define the cross-classified strata (Table~\ref{tab:benchmarks}). The choice reflects three criteria: (1) availability of authoritative, country-level data covering virtually all nations; (2) established relevance to opinion formation in social science; and (3) respondent observability.

\begin{table}[ht]
\centering
\caption{Benchmark data sources.}
\label{tab:benchmarks}
\begin{tabular}{@{}lllrl@{}}
\toprule
Source & Variables & Coverage & Year & Strata \\
\midrule
UN World Population Prospects & Country, Gender, Age & 237 areas & 2023 & 2,699 \\
Pew Global Religious Landscape & Country, Religion & 232 countries & 2010 & 1,607 \\
UN World Urbanization Prospects & Country, Urban/Rural & 233 countries & 2018 & 449 \\
\bottomrule
\end{tabular}
\end{table}

The scorecard also evaluates representativeness at coarser resolutions---Region (22 UN sub-regions) and Continent (6)---crossed with demographic variables, plus marginal distributions. This produces 13 dimensions ranging from Country $\times$ Gender $\times$ Age (most demanding) to Gender (least).

\subsection{Maximum Possible Scores and Efficiency Ratios}

Perfect representativeness (GRI = 1.0) is mathematically impossible at realistic sample sizes when the number of strata is large. We estimate the \textbf{maximum achievable GRI} through Monte Carlo simulation: for each dimension and sample size, we draw 1,000 optimal sample allocations and report the mean maximum GRI (Table~\ref{tab:max_gri}).

\begin{table}[ht]
\centering
\caption{Maximum achievable GRI by sample size (Monte Carlo, 1,000 simulations).}
\label{tab:max_gri}
\begin{tabular}{@{}lccccc@{}}
\toprule
Dimension & $N=100$ & $N=250$ & $N=500$ & $N=1{,}000$ & $N=2{,}000$ \\
\midrule
Country $\times$ Gender $\times$ Age & 0.430 & 0.581 & 0.691 & 0.792 & 0.873 \\
Country $\times$ Religion & 0.721 & 0.839 & 0.898 & 0.938 & 0.965 \\
Country $\times$ Environment & 0.714 & 0.844 & 0.906 & 0.950 & 0.976 \\
\bottomrule
\end{tabular}
\end{table}

The \textbf{efficiency ratio} separates structural limitations from allocation failures:
\begin{equation}\label{eq:efficiency}
    \text{Efficiency} = \frac{\text{GRI}_{\text{actual}}}{\text{GRI}_{\text{max}}}
\end{equation}

The Monte Carlo maximum assumes an oracle allocator; real surveys face recruitment friction, language barriers, and differential willingness to participate. The max GRI is therefore an upper bound on the upper bound---but low efficiency against even this generous ceiling indicates clear room for improvement.

\subsection{The Strategic Representativeness Index (SRI)}

The SRI replaces the population-proportional target with a square-root-proportional target:
\begin{equation}\label{eq:sri_target}
    s_i^* = \frac{\sqrt{q_i}}{\sum_j \sqrt{q_j}}
\end{equation}
\begin{equation}\label{eq:sri}
    \text{SRI} = 1 - \frac{1}{2} \sum_{i=1}^{K} |p_i - s_i^*|
\end{equation}

The $\sqrt{q_i}$ target can be derived as the allocation minimizing the maximum relative estimation error across strata when within-stratum variances are proportional to stratum size. It moderately boosts targets for smaller strata while tempering larger ones---useful for prospective survey design.

\subsection{The Inferential Cost of Low Representativeness}

A natural objection to the GRI is: if the sample's demographic composition doesn't match the population, can't we simply apply post-stratification weights? The answer is yes---but at a quantifiable cost to statistical precision.

When sample proportions $p_i$ differ from population proportions $q_i$, applying post-stratification weights $w_i = q_i / p_i$ produces unbiased population-level estimates. But the variance of the weighted estimator is inflated by the \textbf{design effect} \citep{kish1965survey}:
\begin{equation}\label{eq:deff}
    d_{\text{eff}} = 1 + \text{CV}^2(w)
\end{equation}
where $\text{CV}^2(w) = \text{Var}(w) / \bar{w}^2$ is the squared coefficient of variation of the weights. The \textbf{effective sample size} is:
\begin{equation}\label{eq:neff}
    N_{\text{eff}} = \frac{N}{d_{\text{eff}}}
\end{equation}

A design effect of 3.0 means the weighted estimates have the same precision as a simple random sample one-third the size---two-thirds of the data collection budget has been consumed by the need to reweight.

Equations~\eqref{eq:deff}--\eqref{eq:neff} assume full benchmark coverage: every population stratum contains at least one survey respondent. When some strata are unrepresented ($p_i = 0$, $q_i > 0$), reweighting cannot correct for these gaps---point~(2) in the list below. In practice, we compute $d_{\text{eff}}$ over the \emph{covered} subset of strata, renormalizing both $p$ and $q$ to sum to one over that subset. Let $f = \sum_{i \in \mathcal{C}} p_i$ denote the fraction of survey respondents in benchmark-matched strata $\mathcal{C} = \{i : p_i > 0 \text{ and } q_i > 0\}$. The coverage-adjusted effective sample size is:
\begin{equation}\label{eq:neff_adj}
    N_{\text{eff}} = \frac{N \cdot f}{d_{\text{eff}}}
\end{equation}
where $d_{\text{eff}} = \sum_{i \in \mathcal{C}} \tilde{q}_i^2 / \tilde{p}_i$ uses the renormalized proportions $\tilde{q}_i = q_i / \sum_{j \in \mathcal{C}} q_j$ and $\tilde{p}_i = p_i / f$. When coverage is complete ($f = 1$), this reduces to Equation~\eqref{eq:neff}. All design effect and effective $N$ values reported in this paper use the coverage-adjusted formulation.

The connection to the GRI is through the weight distribution: when $p_i \approx q_i$ for all strata, the weights are near 1.0, the CV of weights is small, and the design effect approaches 1.0. As distributional mismatch grows (GRI decreases), some weights become extreme, the CV increases, and precision degrades. Formally, $d_{\text{eff}} = \sum_i q_i^2 / p_i$---driven by a ratio-based divergence (related to $\chi^2$ distance) rather than the $L_1$ distance underlying TVD. The critical distinction is one of \textbf{symmetry}: TVD treats overrepresentation and underrepresentation of the same absolute magnitude equally, while the design effect is asymmetric---underrepresentation is far more expensive, because the reweighting ratio $q_i/p_i$ amplifies noise from the few respondents in thin strata, whereas overrepresentation merely downweights excess data (wasteful but not precision-destroying). Two surveys with the same GRI can therefore have very different design effects. We recommend reporting \emph{both}: the GRI measures \emph{how much} the sample deviates; the design effect reveals \emph{how costly} that deviation is.

Post-stratification reweighting is not a substitute for representative sampling for four reasons: (1) variance inflation is quadratic---doubling $q_i/p_i$ quadruples that stratum's contribution to the design effect; (2) empty strata ($p_i = 0$) are uncorrectable; (3) extreme weights increase model dependence and sensitivity to outliers; and (4) reweighting treats the symptom (bias) but not the cause (poor allocation), leaving confidence intervals unnecessarily wide.

\section{Empirical Application: The Global Dialogues Case Study}

\subsection{The Global Dialogues Survey}

The Global Dialogues (GD) is a longitudinal survey of public perceptions of AI, recruiting approximately 1,000 participants per wave through online purposive sampling across 50--70 countries (Table~\ref{tab:gd_waves}).

\begin{table}[ht]
\centering
\caption{Global Dialogues survey waves.}
\label{tab:gd_waves}
\begin{tabular}{@{}lrl@{}}
\toprule
Wave & $N$ & Countries \\
\midrule
GD1 & 1,280 & $\sim$60 \\
GD2 & 1,105 & $\sim$55 \\
GD3 & 971 & 63 \\
GD4 & 1,050 & 57 \\
GD5 & 1,057 & $\sim$60 \\
GD6 & 1,037 & $\sim$60 \\
GD7 & 1,022 & 65 \\
\bottomrule
\end{tabular}
\end{table}

\subsection{GRI Scores Across Seven Waves}

Table~\ref{tab:gri_scores} presents GRI scores for primary and selected auxiliary dimensions, and Figure~\ref{fig:scorecard_heatmap} provides a visual summary of the complete 13-dimension scorecard across all waves. Gender balance is essentially perfect (GRI $> 0.97$). Continental representativeness is good (GRI $\approx 0.83$). The picture deteriorates sharply at finer granularity: Country $\times$ Gender $\times$ Age averages only 0.34, indicating roughly 66\% of the joint demographic weight is misallocated.

\begin{table}[ht]
\centering
\caption{GRI scores across Global Dialogues waves GD1--GD7.}
\label{tab:gri_scores}
\begin{tabular}{@{}lccccccccc@{}}
\toprule
Dimension & GD1 & GD2 & GD3 & GD4 & GD5 & GD6 & GD7 & Mean & Range \\
\midrule
Country $\times$ Gender $\times$ Age & 0.337 & 0.329 & 0.348 & 0.361 & 0.347 & 0.341 & 0.330 & 0.342 & 0.03 \\
Country $\times$ Religion & 0.482 & 0.488 & 0.514 & 0.535 & 0.499 & 0.495 & 0.478 & 0.499 & 0.06 \\
Country $\times$ Environment & 0.424 & 0.394 & 0.435 & 0.453 & 0.414 & 0.406 & 0.396 & 0.417 & 0.06 \\
Country & 0.590 & 0.573 & 0.618 & 0.639 & 0.592 & 0.587 & 0.576 & 0.597 & 0.07 \\
Region & 0.745 & 0.739 & 0.791 & 0.799 & 0.738 & 0.734 & 0.744 & 0.756 & 0.06 \\
Continent & 0.832 & 0.830 & 0.886 & 0.883 & 0.773 & 0.802 & 0.796 & 0.829 & 0.11 \\
Religion & 0.817 & 0.819 & 0.833 & 0.826 & 0.813 & 0.806 & 0.788 & 0.814 & 0.05 \\
Gender & 0.989 & 0.990 & 0.996 & 0.979 & 0.986 & 0.995 & 0.996 & 0.990 & 0.02 \\
\bottomrule
\end{tabular}
\end{table}

\begin{figure}[t]
\centering
\includegraphics[width=\textwidth]{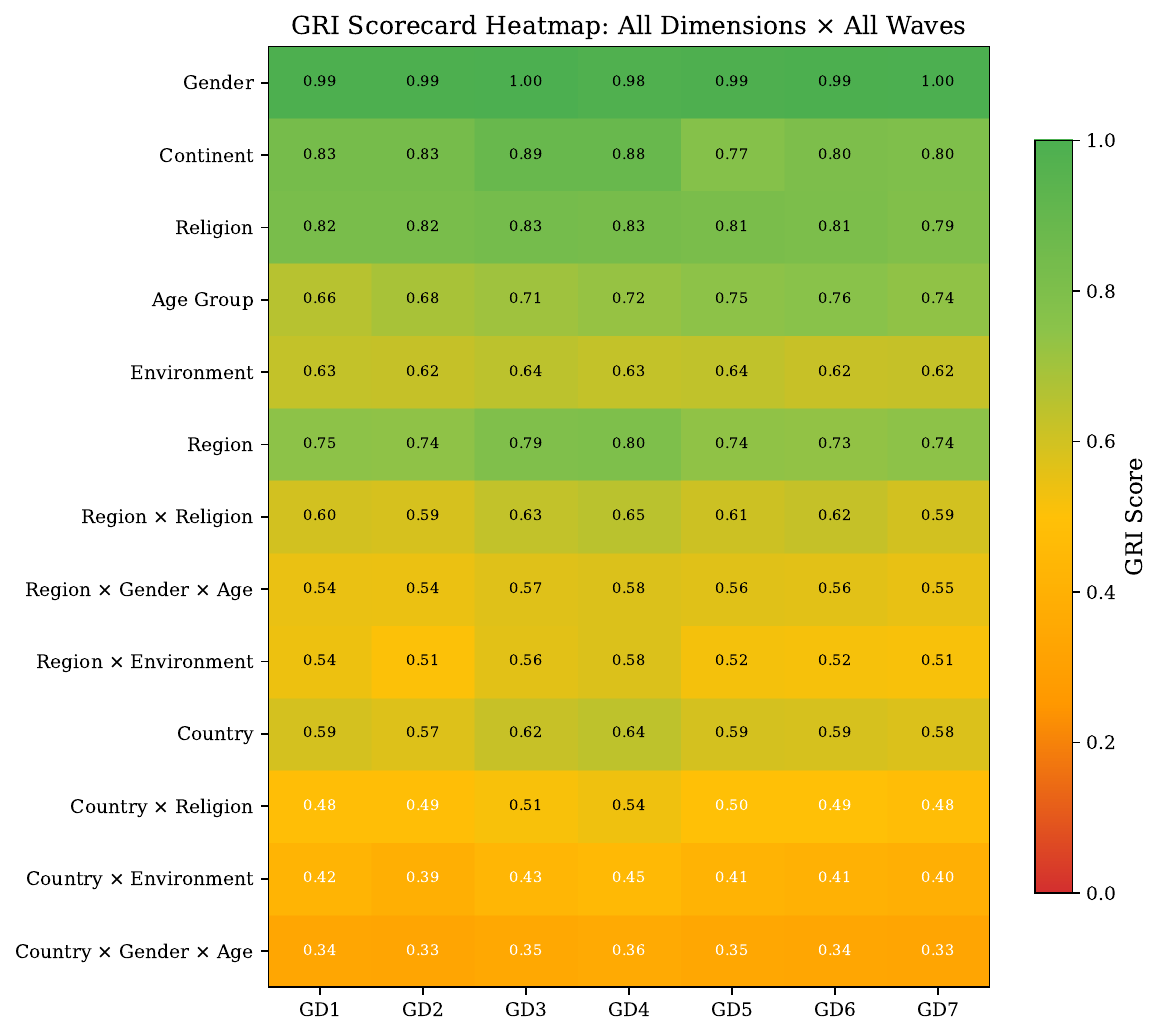}
\caption{GRI scorecard heatmap across all 13 dimensions and Global Dialogues waves. Dimensions are ordered from least demanding (Gender) to most demanding (Country $\times$ Gender $\times$ Age). Color encodes GRI score from poor (red/orange, ${<}0.4$) through moderate (yellow, $0.4$--$0.6$) to excellent (green, ${>}0.8$). The gradient reveals the hierarchical structure of representativeness: marginal demographics are well-captured while fine-grained cross-classifications remain challenging, with scores remarkably stable across waves.}
\label{fig:scorecard_heatmap}
\end{figure}

\subsection{Efficiency Analysis}

Table~\ref{tab:efficiency} shows efficiency ratios---actual GRI divided by Monte Carlo maximum. The GD achieves approximately 43\% efficiency on Country $\times$ Gender $\times$ Age and 53\% on Country $\times$ Religion, indicating room for improvement through better sampling design. GD4 achieves the highest efficiency across all primary dimensions.

\begin{table}[ht]
\centering
\caption{Efficiency ratios (actual GRI / max GRI) for primary dimensions.}
\label{tab:efficiency}
\begin{tabular}{@{}lccccccccl@{}}
\toprule
Dimension & Max & GD1 & GD2 & GD3 & GD4 & GD5 & GD6 & GD7 & Mean \\
\midrule
C $\times$ G $\times$ A & 0.792 & 42.5\% & 41.6\% & 44.0\% & 45.5\% & 43.8\% & 43.0\% & 41.7\% & 43.2\% \\
C $\times$ Religion & 0.938 & 51.4\% & 52.0\% & 54.8\% & 57.0\% & 53.2\% & 52.7\% & 50.9\% & 53.1\% \\
C $\times$ Environ. & 0.950 & 44.7\% & 41.5\% & 45.7\% & 47.7\% & 43.6\% & 42.7\% & 41.6\% & 43.9\% \\
\bottomrule
\end{tabular}
\end{table}

\subsection{Metric Comparison: GRI, SRI, and Efficiency}

Table~\ref{tab:variants} compares GRI and SRI for GD5 alongside the efficiency ratio. SRI scores are generally similar to GRI but differ at extremes: for Continent, SRI is higher (0.841 vs.\ 0.773) because the square-root transformation boosts smaller continents; for Religion, SRI is lower (0.745 vs.\ 0.813) because strategic targets boost underrepresented religions.

\begin{table}[ht]
\centering
\caption{Comparison of GRI, SRI, and efficiency ratio for GD5 ($N = 1{,}057$).}
\label{tab:variants}
\begin{tabular}{@{}lcccc@{}}
\toprule
Dimension & GRI & SRI & Max GRI & Efficiency \\
\midrule
Country $\times$ Gender $\times$ Age & 0.347 & 0.381 & 0.792 & 43.8\% \\
Country $\times$ Religion & 0.499 & 0.447 & 0.938 & 53.2\% \\
Country $\times$ Environment & 0.414 & 0.387 & 0.950 & 43.6\% \\
Country & 0.592 & 0.506 & --- & --- \\
Region & 0.738 & 0.749 & --- & --- \\
Continent & 0.773 & 0.841 & --- & --- \\
Religion & 0.813 & 0.745 & --- & --- \\
Gender & 0.986 & 0.986 & --- & --- \\
\bottomrule
\end{tabular}
\end{table}

\subsection{The Inferential Cost: Design Effect and Effective Sample Size}

Table~\ref{tab:deff} presents the design effect and effective sample size for the three primary dimensions across all GD waves.

\begin{table}[ht]
\centering
\caption{Design effect ($d_{\text{eff}}$), effective $N$ (Equation~\ref{eq:neff_adj}), and precision retained for GD waves on Country $\times$ Gender $\times$ Age.}
\label{tab:deff}
\begin{tabular}{@{}lrrrr@{}}
\toprule
Wave & $N$ & $d_{\text{eff}}$ & $N_{\text{eff}}$ & Precision \\
\midrule
GD1 & 1,280 & 3.18 & 199 & 31.4\% \\
GD2 & 1,105 & 3.85 & 136 & 25.9\% \\
GD3 & 971 & 2.95 & 166 & 33.9\% \\
GD4 & 1,050 & 2.46 & 233 & 40.7\% \\
GD5 & 1,057 & 3.02 & 195 & 33.1\% \\
GD6 & 1,037 & 2.69 & 202 & 37.2\% \\
GD7 & 1,022 & 3.09 & 188 & 32.3\% \\
\midrule
\textbf{Mean} & \textbf{1,075} & \textbf{3.04} & \textbf{189} & \textbf{33.5\%} \\
\bottomrule
\end{tabular}
\end{table}

On average, the GD's $\sim$1,000 respondents yield an effective sample size of only \textbf{189} for Country $\times$ Gender $\times$ Age inference---roughly two-thirds of the data budget is consumed by reweighting. For Country $\times$ Environment, average $d_{\text{eff}} = 11.41$ yields $N_{\text{eff}} \approx 96$. The design effect table reveals the asymmetry discussed in Section~3.6: Country $\times$ Environment has the highest design effects despite fewer strata, reflecting severe rural underrepresentation in online surveys.

\subsection{Cross-Survey Validation: The World Values Survey}\label{sec:wvs}

We apply the framework to seven waves of the World Values Survey \citep{haerpfer2022wvs}, the largest cross-national values survey, employing probability sampling with 10--60 countries per wave and $N = 9{,}144$--$85{,}219$ (Table~\ref{tab:wvs_comparison}).

\begin{table}[ht]
\centering
\caption{GRI comparison: Global Dialogues vs.\ World Values Survey.}
\label{tab:wvs_comparison}
\begin{tabular}{@{}lcc@{}}
\toprule
Metric & GD (7-wave mean) & WVS (7-wave mean) \\
\midrule
Sample size per wave & $\sim$1,075 & $\sim$57,600 \\
Countries per wave & $\sim$59 & $\sim$40 \\
CGA GRI & 0.342 & 0.195 \\
Country $\times$ Religion GRI & 0.499 & 0.301 \\
Country $\times$ Environment GRI & 0.417 & 0.320 \\
CGA $d_{\text{eff}}$ & 3.04 & 2.71 \\
CGA $N_{\text{eff}}$ & 189 & 5,695 \\
CGA Precision Retained & 33.5\% & 42.6\% \\
\bottomrule
\end{tabular}
\end{table}

The comparison reveals a fundamental trade-off: the GD achieves higher GRI because it recruits from more countries ($\sim$58 vs.\ $\sim$40), while the WVS achieves higher effective $N$ through sheer sample size despite worse global representativeness. This demonstrates precisely why reporting \emph{both} GRI and effective $N$ is essential: GRI alone would favor the GD; effective $N$ alone would favor the WVS; together they communicate the complete picture.

\subsection{Regional Surveys: Afrobarometer and Latinobar\'{o}metro}

We also apply the framework to Afrobarometer Round~9 \citep{afrobarometer2024} ($N = 53{,}444$, 39~countries) and Latinobar\'{o}metro 2023--2024 \citep{latinobarometro2024} ($N \approx 19{,}200$, 17~countries), using country-filtered benchmarks restricted to each survey's claimed coverage area. Afrobarometer achieves CGA GRI = 0.532 with $d_{\text{eff}} = 2.01$ ($N_{\text{eff}} = 16{,}152$); Latinobar\'{o}metro 2023 achieves CGA GRI = 0.480 with $d_{\text{eff}} = 2.51$ ($N_{\text{eff}} = 5{,}666$). These regional surveys demonstrate that large-scale probability surveys with comprehensive regional coverage can achieve good representativeness combined with strong inferential power. The year-over-year stability of Latinobar\'{o}metro scores (2023 vs.\ 2024) also demonstrates the GRI's utility for longitudinal monitoring.

\subsection{The Complementarity of GRI and Effective Sample Size}

Figure~\ref{fig:cross_survey} plots GRI against effective sample size for all five surveys across the three primary benchmark dimensions. The clustering reveals the central empirical finding: no survey dominates on both metrics simultaneously. The Global Dialogues achieves moderate GRI through broad country coverage but yields low $N_{\text{eff}}$ due to small per-country samples. The World Values Survey achieves high $N_{\text{eff}}$ through large within-country samples but low GRI because its country selection covers a fraction of the global demographic distribution. The regional surveys---evaluated against country-filtered benchmarks for their claimed coverage areas---achieve both high GRI and high $N_{\text{eff}}$, illustrating the advantage of matching sampling ambition to achievable scope. This figure provides visual evidence for the paper's recommendation: reporting GRI and $N_{\text{eff}}$ together communicates what neither metric captures alone.

\begin{figure}[ht]
\centering
\includegraphics[width=\textwidth]{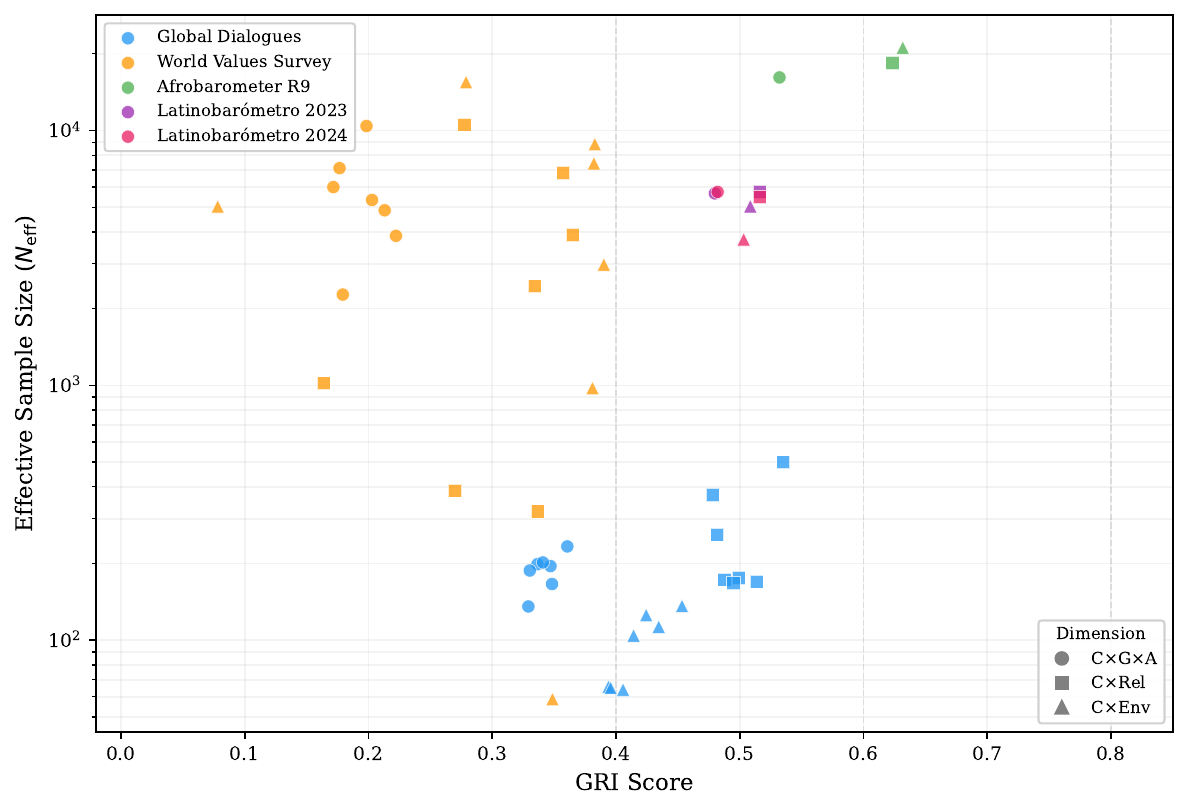}
\caption{GRI vs.\ effective sample size across five surveys and three primary benchmark dimensions (Country $\times$ Gender $\times$ Age, Country $\times$ Religion, Country $\times$ Environment). Each point represents one survey wave evaluated on one dimension. Regional surveys (Afrobarometer, Latinobar\'{o}metro) use country-filtered benchmarks; global surveys (GD, WVS) use world population benchmarks. The log-scale $y$-axis spans three orders of magnitude, reflecting the vast differences in inferential power across survey designs.}
\label{fig:cross_survey}
\end{figure}

\section{What Counts as Representative?}

\subsection{The Interpretation Challenge}

A GRI of 0.34 for Country $\times$ Gender $\times$ Age sounds alarming. Three considerations moderate this. First, high GRI on fine-grained dimensions is \emph{mathematically constrained} at moderate sample sizes---even a perfectly allocated sample of 1,000 can achieve at most 0.792. The application to the World Values Survey (Section~\ref{sec:wvs}) confirms this is universal: the WVS, with 60,000--85,000 respondents per wave, achieves CGA GRI scores of only 0.17--0.22 due to limited country coverage. Second, the GRI measures distributional fidelity, not fitness for purpose: a survey designed to compare AI attitudes between specific countries might intentionally oversample them. Third, the efficiency ratio is more informative than the raw score---43\% efficiency directly suggests room for improvement.

\subsection{Why Religion Scores Higher Than Demographics}

Country $\times$ Religion GRI (mean: 0.50) exceeds Country $\times$ Gender $\times$ Age (mean: 0.34) for three reasons: fewer strata (1,607 vs.\ 2,699), larger stratum proportions (major religions cover $>85\%$ of the population), and geographic correlation (covering many countries naturally captures religious diversity). Country $\times$ Environment (mean: 0.42), despite having the fewest strata (449), scores lower due to systematic urban bias in online surveys.

\subsection{Beyond Surveys: ML Dataset Auditing}

The GRI framework applies to any dataset with categorical demographic attributes and a reference population. Training datasets with demographic skews produce models that perform differently across groups \citep{buolamwini2018gender}. Current documentation practices---Datasheets for Datasets \citep{gebru2021datasheets}, Model Cards \citep{mitchell2019model}---advocate qualitative demographic description. The GRI provides the quantitative counterpart: a standardized score enabling comparison across datasets and tracking over time.

\section{Practical Implementation}

\subsection{The \texttt{gri} Python Library}

The framework is implemented as an open-source Python library comprising 16 modules organized into calculation, analysis, data, and presentation layers.\footnote{Source code and documentation: \url{https://github.com/collect-intel/gri}} It exposes two API surfaces:

\begin{lstlisting}[caption={Functional API for scripted pipelines.},label={lst:functional}]
from gri import calculate_gri, load_benchmark_suite

benchmarks = load_benchmark_suite()
gri_score = calculate_gri(
    survey_df, benchmarks['country_gender_age'],
    strata_cols=['country', 'gender', 'age_group']
)
\end{lstlisting}

\begin{lstlisting}[caption={Object-oriented API for interactive analysis.},label={lst:oo}]
from gri import GRIAnalysis

analysis = GRIAnalysis.from_survey_file('survey.csv')
scorecard = analysis.calculate_scorecard(
    include_max_possible=True, n_simulations=1000
)
analysis.plot_scorecard(save_to='scorecard.png')
\end{lstlisting}

\section{Limitations and Future Work}

\textbf{Benchmark staleness.} The religious composition data dates to 2010. While religious demographics shift slowly in most regions, rapid secularization in several countries means the benchmark may overstate affiliated populations. We recommend reporting benchmark vintage alongside scores and updating on a 5-year cycle.

\textbf{Equal weighting.} The framework reports separate dimension scores without prescribing aggregation weights. A health survey might weight age/gender higher; a governance study might weight country higher. We deliberately leave this to researchers.

\textbf{Continuous variables.} Age is discretized into 6 brackets, losing within-bracket distributional information. Extension to continuous demographics would require kernel-density-based distance metrics.

\textbf{Intersectionality.} The GRI evaluates pre-specified dimension crosses but cannot capture intersectionalities not encoded in the strata. The framework is extensible---any categorical variable with a benchmark can be added---but dimension selection has equity implications.

\textbf{Diversity Score non-monotonicity.} The $1/N$ threshold means that as $N$ increases, more strata become ``relevant.'' The Diversity Score can decrease with larger samples if newly relevant strata go unrepresented---semantically correct but potentially counterintuitive.

\textbf{Benchmark uncertainty.} Population benchmarks are estimates with uncertainty. A fully Bayesian treatment propagating benchmark uncertainty into GRI confidence intervals---treating $q_i$ as random variables reflecting census quality---remains future work.

\textbf{Scalability.} Adding dimensions multiplicatively increases strata. The pairwise scorecard approach captures the most important signals but loses higher-order interaction information.

\textbf{Subnational variation.} The current library ships only with global benchmarks. Integrating subnational data from national censuses would substantially broaden practical utility.

\section{Conclusion}

The GRI transforms a vague question---``How representative is this survey?''---into a precise, measurable quantity. Application to seven waves of the Global Dialogues survey reveals that a purposive online survey spanning 60+ countries achieves Country $\times$ Gender $\times$ Age GRI scores of only 0.33--0.36, capturing roughly 43\% of the theoretically achievable representativeness. The primary drivers are geographic concentration and structural urban bias. Cross-survey validation using the World Values Survey (seven waves, 9{,}000--85{,}000 respondents per wave) reveals a complementary pattern: WVS achieves lower GRI scores (mean CGA GRI $= 0.20$) despite far larger samples, because its country selection covers fewer of the world's demographic strata---but its effective sample sizes ($N_{\text{eff}} \approx 5{,}700$) dwarf those of the GD waves ($N_{\text{eff}} \approx 189$), reflecting the inferential advantage of probability sampling within covered populations. Application to the Afrobarometer (Round~9, $N = 53{,}444$) and Latinobar\'ometro (2023--2024, $N \approx 19{,}200$) further demonstrates generalizability across survey modalities, scales, and geographic scopes.

The GRI and design effect measure complementary aspects of sample quality, and we recommend that surveys report both as a minimum viable summary. The \textbf{GRI} answers ``How closely does this sample mirror the target population?''---a symmetric measure that treats all demographic deviations equally. The \textbf{effective sample size} (Equation~\ref{eq:neff_adj}) answers ``What is this survey actually worth for inference?''---an asymmetric measure that penalizes underrepresentation far more heavily than overrepresentation, because reweighting amplifies noise from few respondents. Together, these two numbers communicate what neither can alone: a survey with GRI $= 0.34$ and $N_{\text{eff}} = 189$ (like the typical GD wave) tells a very different story from one with GRI $= 0.20$ and $N_{\text{eff}} = 5{,}695$ (like the typical WVS wave), even though both are imperfect.

What changes if researchers adopt this framework? Three things. First, \textbf{transparency}: every survey can report a standardized, comparable representativeness score alongside its results, enabling consumers of research to calibrate their confidence in ``global'' findings. Second, \textbf{optimization}: the segment-level decomposition and efficiency analysis identify exactly which demographic gaps matter most, guiding recruitment strategy. Third, \textbf{accountability}: funders and policymakers can set minimum representativeness targets for research that claims to represent global populations, just as they set minimum sample size requirements.

The complementary metrics extend this toolkit: the SRI for designing surveys that maximize statistical power across populations, the efficiency ratio for contextualizing scores against structural constraints, the diversity score for measuring strata coverage, and the design effect for quantifying the inferential cost of distributional mismatch. Together, these metrics span the lifecycle of survey research---design (SRI, Monte Carlo max scores), execution (real-time GRI monitoring, segment deviation analysis), and evaluation (GRI, efficiency ratio, design effect, effective sample size). The framework is released as open-source software with authoritative population benchmarks from the United Nations and Pew Research Center. We invite the survey methodology community to adopt, critique, and extend it.

\section*{Acknowledgments}

We thank the Collective Intelligence Project team, especially Zarinah Agnew, for developing the Global Dialogues project, survey instruments, and partnerships, and Joal Stein for maintaining the Global Dialogues website and public data access. We thank the Prolific team for their collaboration on globally representative sampling across each Global Dialogue wave. We also thank the World Values Survey Association, Afrobarometer, and Corporaci\'{o}n Latinobar\'{o}metro for making their survey data publicly available.

\section*{Code and Data Availability}

The \texttt{gri} Python library, including all benchmark data from the United Nations and Pew Research Center, is available at \url{https://github.com/collect-intel/gri}. Complete scorecard tables for all surveys analyzed in this paper are included in the repository and can be reproduced via \texttt{make reproduce}. Global Dialogues survey data is available via the repository as a git submodule. World Values Survey data is available from \url{https://www.worldvaluessurvey.org}. Afrobarometer data is available from \url{https://www.afrobarometer.org}. Latinobar\'{o}metro data is available from \url{https://www.latinobarometro.org}.

\bibliographystyle{plainnat}
\bibliography{references}

\appendix

\section{Mathematical Proofs}\label{app:proofs}

\subsection{Proof of GRI Bounds}

See Theorem~\ref{thm:bounds}. The lower bound TVD $= 0$ is achieved when $P = Q$. The upper bound TVD $= 1$ is achieved when $P$ and $Q$ have disjoint support. $\square$

\subsection{Relationship Between TVD and Other Distances}

TVD relates to the Hellinger distance $H$ and KL divergence $D_{\text{KL}}$:
\begin{align}
    H^2(P,Q) &\leq \text{TVD}(P,Q) \leq H(P,Q)\sqrt{2} \label{eq:hellinger_bound} \\
    \text{TVD}(P,Q) &\leq \sqrt{\tfrac{1}{2} D_{\text{KL}}(P \| Q)} \label{eq:pinsker}
\end{align}
Equation~\eqref{eq:pinsker} is Pinsker's inequality \citep{lecam1986asymptotic,tsybakov2009introduction}.

\section{Pseudocode}\label{app:pseudocode}

\begin{algorithm}
\caption{GRI Calculation}\label{alg:gri}
\begin{algorithmic}[1]
\Require Survey data $S$, benchmark $B$, strata columns $C$
\State $\mathbf{p} \gets \text{GroupAndCount}(S, C) / |S|$
\State $\mathbf{q} \gets \text{Normalize}(B.\text{proportion})$
\State $M \gets \text{OuterJoin}(\mathbf{p}, \mathbf{q}, C)$; fill missing with 0
\State $\text{TVD} \gets 0.5 \times \sum_i |M.p_i - M.q_i|$
\State \Return $1.0 - \text{TVD}$
\end{algorithmic}
\end{algorithm}

\begin{algorithm}
\caption{Monte Carlo Maximum GRI}\label{alg:mc}
\begin{algorithmic}[1]
\Require Benchmark proportions $\mathbf{q}$, sample size $N$, iterations $T$
\For{$t = 1$ to $T$}
    \State $\mathbf{n}^* \gets \text{round}(\mathbf{q} \times N)$ for large strata; Bernoulli for small
    \State Adjust total to $N$
    \State $\text{GRI}_t \gets 1 - 0.5 \times \sum_i |n^*_i/N - q_i|$
\EndFor
\State \Return $\text{mean}(\{\text{GRI}_t\})$, $\text{std}(\{\text{GRI}_t\})$
\end{algorithmic}
\end{algorithm}

\section{Extended Scorecard Tables}\label{app:tables}

Figure~\ref{fig:scorecard_heatmap} provides a visual summary of GRI scores across all dimensions and waves. Complete scorecard tables for all 13 dimensions---including GRI, Diversity Score, SRI, design effect, effective $N$, and precision retention---are available in the project repository (\texttt{analysis\_output/scorecards/}) and can be reproduced with \texttt{make reproduce}.

\end{document}